\documentclass[a4paper]{amsart}
\usepackage{a4wide}

\usepackage{amsmath,amsthm}
\usepackage{amssymb}

\usepackage{graphicx}

\newtheorem{theorem}{Theorem}[section]

\newtheorem{lemma}[theorem]{Lemma}

\theoremstyle{definition}
\newtheorem{definition}[theorem]{Definition}
\newtheorem{example}[theorem]{Example}

\renewcommand{\d}{\mathrm{d}}

\newcommand{\EL}{\mathcal{EL}}

\title[Covariant Hamiltonian dynamics with constraints]{Covariant Hamiltonian first order field theories with constraints on manifolds with boundary: the case of Hamiltonian dynamics}

\author{A. Ibort}
\address{ICMAT and Depto. de Matem\'aticas, Univ. Carlos III de Madrid, \\ Avda. de la
Universidad 30, 28911 Legan\'es, Madrid, Spain.\\
E-mail: albertoi@math.uc3m.es}

\author{A. Spivak}
\address{Dept. of Mathematics, Univ. of California at Berkeley, \\ 903 Evans Hall, 94720 Berkeley CA, USA.\\
E-mail: ameliaspivak@gmail.com}


\begin{document}

\begin{abstract}
Inspired by problems arising in the geometrical treatment of Yang-Mills theories and Palatini's gravity, the covariant formulation of Hamiltonian dynamical systems as a Hamiltonian field theory of dimension $1+0$ on a manifold with boundary is discussed.   After a precise statement of Hamilton's variational principle in this context, the geometrical properties of the space of solutions of the Euler-Lagrange equations of the theory are analyzed. A sufficient condition is obtained that guarantees that the set of solutions of the Euler-Lagrange equations at the boundary of the manifold, fill a Lagrangian submanifold of the space of fields at the boundary.  Finally a theory of constraints is introduced that mimics the constraints arising in Palatini's gravity.   
\end{abstract}

\maketitle

\section{Introduction}\label{sec:introduction}
In this paper we analyze in depth the theory of Hamiltonian dynamical systems, viewing them as first-order covariant Hamiltonian field theories on a manifold of dimension 1+0 with boundary. The reason for our interest in such simple field theories is twofold: First, it allows us to discuss in a simple setting two features which are relevant to two fundamental examples of first-order covariant Hamiltonian field theories, Yang-Mills and especially Palatini gravity.
These features are the introduction of constraints and the 'topological phases' of the theory. Secondly, it allows us to test some common assumptions, among them the role of boundary conditions and the geometry of the restriction to the boundary of the space of solutions of the Euler-Lagrange equations.

In the recent paper \cite{Ib15} the multisymplectic approach to first order Hamiltonian field theories on manifolds with boundary is discussed and various examples are treated, among them Yang-Mills theories.    It is natural to consider Palatini's gravity as a first order Hamiltonian field theory and, in a precise sense as a `topological phase' of a gauge theory, that is, as a limit of a gauge theory where the kinetic term of the theory disappears.   Moreover, a constraint in the momenta fields of the theory must be introduced to recover the equations of motion of Palatini's gravity, such constraint relates the momenta to the vierbein fields $e_\mu$ of the standard treatment. 

Thus we will exhibit a simple instance of a similar program but using instead Hamiltonian dynamics considered as a covariant first order Hamiltonian formulation of field theories on manifolds with boundary in dimension $1+0$ as discussed in \cite{Ib15}.   This theory provides a multisymplectic background to the quantum perturbative programme  discussed in \cite{Ca11,Ca14}.  The role of a covariant phase space for a first order Hamiltonian theory modelled on the affine dual space of the first jet bundle of the bundle defining the fields of the theory is assessed and the crucial role played by boundaries as determining symplectic spaces of fields defining the classical counterpart of the quantum states of the theory is stressed in accordance with the point of view expressed in \cite{Sc51}.    The fact that the basic space of the theory is one-dimensional and that the Euler-Lagrange equations of the theory are the standard Hamilton's equations, removes most of the analytical difficulties arising in higher-order theories. It is possible to show that, under appropriate conditions, the spaces of fields have specific geometric structures.  Most notably, it will be shown that for theories named locally Dirichlet, the space of solutions of the Euler-Lagrange equations restricted to the boundary is a Lagrangian submanifold and a locally defined generating function will be constructed.     

Section \ref{sec:general} will be devoted to discuss the formulation of Hamiltonian dynamics as a first order Hamiltonian field theory on a manifold with boundary and the main ingredients of the theory will be established: the covariant phase space, the symplectic structure on the space of fields at the boundary, the action and Hamilton's variational principle and the Euler-Lagrange equations of the theory.   In Section \ref{sec:boundary} the problem of determining under what conditions the space of solutions of Euler-Lagrange equations at the boundary is Lagrangian will be addressed.   A solution will be provided by exhibiting the construction of locally defined generating functions and a number of examples will be discussed.   Finally, in Section \ref{sec:constraints} the introduction of constraints in the theory will be analyzed.   In particular constraints affecting the momenta of the theory will be considered and a reduction theory for them will be obtained.


\section{The geometry of the covariant phase space for Hamiltonian dynamics}\label{sec:general}

  We will quickly review the basic notions and notations for first order covariant Hamiltonian field theories. We will stick to the notations and definitions in \cite{Ib15},for more background details see also \cite{Ca91}. Using Eqs.(3) and(9)we will then be able to prove Theorem 2.1, for all first order covariant Hamiltonian theories.

\subsection{The covariant phase space of first order Hamiltonian field theories}\label{sec:covariant}

The fundamental geometrical structure of a first order covariant Hamiltonian field theory will be provided by a fiber bundle $\pi \colon E \to M$ with $M$ a $m = (1+d)$-dimensional orientable smooth manifold with smooth boundary $\partial M \neq \emptyset$ and local coordinates adapted to the fibration $(x^\mu, u^a)$, $a= 1, \ldots, r$, where $r$ is the dimension of the standard fiber.     Because $M$ is orientable we will assume that a given volume form $\mathrm{vol}_M$ is selected.  Notice that it is always possible to chose local coordinates $x^\mu$ such that $\mathrm{vol}_M = dx^0 \wedge dx^1 \wedge \cdots \wedge dx^d$.   

In this setting, the formalism for Hamiltonian dynamics as a covariant Hamiltonian field theory, will be provided by  the bundle $\pi \colon E = Q \times [0,1] \to [0,1]$, where $\pi$ denotes the standard projection on the second factor, $Q$ is an $r$-dimensional smooth manifold without boundary,  and the spacetime of dimension $m = 1+0$, $M = [0,1] \subset \mathbb{R}$, with standard metric $\eta = -d t^2$.  The volume form is given by $\mathrm{vol}_M = dt$.  A standard bundle chart will have the form $(t,u^a)$ where $u^a$, $a = 1, \ldots, r$,  is a local chart for $Q$.   

We will denote by $\pi_1^0\colon J^1E \to E$ the affine 1-jet bundle of the bundle $E\stackrel{\pi}{\rightarrow}M$.
The elements of $J^1E$ are equivalence classes of germs of sections $\phi$ of $\pi$.   The bundle $J^1E$ is an affine bundle over $E$ modeled on the vector bundle $VE \otimes_E \pi^* (T^*M)$.   If $(x^\mu; u^a)$, $\mu = 0, \ldots, d$ is a bundle chart for the bundle $\pi\colon E\to M$, then we will denote by $(x^\mu,u^a; u^a_\mu)$ a local chart for the jet bundle $J^1E$.  By identifying the space of affine functions on each fiber of $J^1E$  with a  space of $m$-forms we can regard $J^1E^*$ as a quotient of a bundle of $m$-forms over $E$.

The affine dual bundle $J^1E^*$ of $J^1E$ is defined as the canonical affine bundle over $E$ modeled on $\pi^* (TM) \otimes_E VE^* $ with local coordinates $(x^\mu,u^a; \rho_a^\mu)$ and canonical projection $\tau_1^0 \colon J^1E^* \to E$ consisting on affine maps on $J^1E \to E$.   In the previous instance it is clear that the  first jet bundle $J^1E$ is canonically isomorphic to $TQ\times [0,1]$ and its affine dual $J^1E^*$ is isomorphic to $T^*Q\times [0,1]$.   Bundle coordinates on $T^*Q\times [0,1]$ will be denoted as $(t, u^a, p_a)$ and the projection $\tau_1^0	\colon J^1E^* \to E$ above, becomes $\tau_1^0(t,u^a, p_a) = (t,u^a)$.

Let $\Lambda_1^m (E)$ be the subbundle of $\Lambda^mE$ of $m$--forms on $E$ consisting of those $m$--forms which vanish when two of their arguments are vertical, then we have a short exact sequence of affine bundles:
$$0\rightarrow\textstyle{\bigwedge}^m_0E\hookrightarrow
	\textstyle{\bigwedge}^m_1E\rightarrow J^1E^*\rightarrow0.  $$
Here $\textstyle{\bigwedge}_0^mE$ is the bundle of basic $m$-forms on $E$, hence the bundle $\textstyle{\bigwedge}_1^mE$ is an affine real line bundle over $J^1E^*$.    Given coordinates $(x^\mu,u^a)$ for $E$ we have coordinates $(x^\mu,u^a,\rho, \rho^\mu_a)$ on $\bigwedge^m_1E$ adapted to them. 	The point $\varpi\in\bigwedge^m_1E$ with coordinates $(x^\mu,u^a;\rho, \rho^\mu_a )$ is the $m$-covector
$$
\varpi =  \rho^\mu_a\, d u^a\wedge \mathrm{vol}_\mu + \rho \, \mathrm{vol}_M \, ,
$$
where $\mathrm{vol}_\mu$ stands for $i_{\partial/\partial x^\mu} \mathrm{vol}_M$.

The bundle $\textstyle{\bigwedge}_1^mE$ carries a canonical $m$--form  which may be defined by a generalization of the definition of the canonical 1--form on the cotangent bundle of a manifold.  Let $\sigma \colon \textstyle{\bigwedge}_1^mE \to E$  be the canonical projection, then the canonical $m$-form $\Theta$ is defined by 
$$
\Theta_\varpi(U_1,U_2,\ldots,U_m)=\varpi(\sigma_*U_1, \ldots, \sigma_*U_m) \, ,
$$
where $\varpi\in\bigwedge^m_1E$ and $U_i\in T_\varpi(\bigwedge^m_1E)$. The closed $(m+1)$-form $\d\Theta$  defines a multisymplectic structure on $\bigwedge^m_1E$.   We shall denote the projection $\bigwedge^m_1E \to E$ by $\nu$, while the projection $\bigwedge^m_1E \to J_1E^*$ will be denoted by $\mu$.  Thus $\nu = \tau^0_1\circ\mu$, with $\tau_1^0 \colon J^1E^* \to E$ the canonical projection.

Notice that the space of 1-semibasic forms on $E=Q\times [0,1] \to [0,1]$ is just the space of 1-forms on $Q\times [0,1]$, that is $\bigwedge^m_1E = T^*(Q\times [0,1])$ and the projection map $\mu \colon \bigwedge^m_1E \to J^1E^*$ is just the projection $\mu \colon T^*(Q\times [0,1]) \to T^*Q\times [0,1]$, given as $\mu(t,u^a; p_0,p_a) = (t,u^a,p_a)$ where $p_adu^a + p_0 dt$ is a generic element in $T^*(Q\times [0,1])$.  The canonical 1-form $\Theta$ on $(\bigwedge^m_1E) \cong T^*(Q\times [0,1])$ is just the canonical Liouville 1-form on the cotangent bundle $T^*(Q\times [0,1])$, and has the expression:
$$
\Theta = p_a du^a + p_0 dt \, .
$$		
	
We define a Hamiltonian $H$ on $J^1E^*$ to be a section of $\mu$, and we can use a Hamiltonian section $\rho = -H(x^\mu, u^a, \rho_a^\mu )$ to define an $m$-form on $J^1E^*$
by pulling back the canonical $m$-form $\Theta$ from $\bigwedge^m_1E$.   We call the form so obtained the Hamiltonian $m$-form associated with $H$ and denote it by $\Theta_H$; thus
\begin{equation}\label{ThetaH}
	\Theta_H = \rho_a^\mu d u^a \wedge \mathrm{vol}_\mu - H(x^\mu,u^a, \rho_a^\mu) \, \mathrm{vol}_M .
\end{equation}
and the pair $(J^1E^*, \Theta_H)$ will be called the Hamiltonian covariant phase space of the theory.

In the simple example in dimension $1+0$ we are considering, a Hamiltonian, i.e., a section of $\mu$, can be identified with a map $H\colon T^*Q \times I\to \mathbb{R}$, and $p_0 = - H(t,u^a, p_a)$, because the bundle defined by the projection $\mu$ is trivial.  The pull-back of the canoncial 1-form $\Theta$ to $J^1E^* \cong T^*Q \times [0,1]$ becomes the standard 1-form of Hamiltonian dynamics:
$$
\Theta_H = p_a du^a - H(t,u,p) dt \, .
$$


\subsection{The action and the variational principle}\label{sec:variational}
	
The fields of the theory are double sections of the bundle $J^1E^*$, that is, sections $\chi\colon M \to J^1 E^*$ of the bundle structure over $M$ such that both $\Phi = \tau_1^0 \circ \chi \colon M \to E$ and $P= \chi\circ \pi \colon E \to J^1E^*$ are sections of $\pi \colon E\to M$	and $\tau_1^0 \colon J^1E^* \to E$ respectively.  Notice that in such a case $P\circ \Phi = \chi$.   We will denote such a section $\chi$ by $(\Phi, P)$ to indicate the double bundle structure  of $J^1E^*$.
 
In the case of Hamiltonian dynamics viewed as a field theory over $M = [0,1]\subset \mathrm{R}$, double sections of the bundle $J^1E^* = T^*Q \times [0,1] \to E=Q\times [0,1] \to [0,1]$, have the form $\chi = P \circ \Phi$, with $\Phi (t) = (u(t),t)$ and $P(u,p,t) = (u,p(t),t)$, i.e., therefore $\chi (t) = (u(t), p(t), t)$.  Thus the space of fields of the theory can be described equivalently as the space of smooth curves on $T^*Q$.

We will denote by $\mathcal{F}_M(E)$, or just $\mathcal{F}_M$ if there is no risk of confusion, the sections $\Gamma (E)$ of the bundle $E$, that is $\Phi \in \mathcal{F}_M$, and by $J^1\mathcal{F}_M^*$ the double sections $\chi = (\Phi, P)$.  Thus $J^1\mathcal{F}_M^*$  represents the space of fields of the theory in the first order covariant Hamiltonian formalism.  In our particular instance $\mathcal{F}_M \cong C^\infty([0,1],Q)$ and $J^1\mathcal{F}_M^* \cong C^\infty([0,1], T^*Q)$. 

The equations of motion of the theory will be defined by means of a variational principle, i.e., they will be characterized as the critical points of an action functional $S$ on $J^1\mathcal{F}_M^*$. Such action will be simply given by:
\begin{equation}\label{action}
 S(\chi ) = \int_M \chi^*\Theta_H \, ,
 \end{equation}
that in our case of Hamiltonian dynamics viewed as a field theory on $M=[0,1]\subset \mathrm{R}$ becomes,
\begin{equation}\label{S_mechanics}
S(\chi) = \int_0^1 \left( p_a(t) \dot{u}^a(t) - H(t,u(t),p(t))  \right) \, dt \, ,
\end{equation}
which is just the standard functional in Hamilton's variational principle.

A simple computation (see \cite{Ib15} for details) leads us to:
\begin{equation}\label{dSfirst}
\mathrm{d} S (\chi) (U) = \int_M \chi^* \left(i_{\widetilde U} d\Theta_H \right) + \int_{\partial M} (\chi\circ i)^* \left(i_{\widetilde U} \Theta_H \right) \, , 
\end{equation}
where $U$ is a vector field on $J^1E^*$ along the section $\chi$, $\widetilde{U}$ is any extension of $U$ to a tubular neighborhood of the graph of $\chi$, and $i\colon \partial M \to M$ is the canonical embedding.


\subsection{The cotangent bundle of fields at the boundary}\label{sec:cotangent_boundary}

Consider a collar around the boundary $U_\epsilon \cong (-\epsilon, 0] \times \partial M$, and local coordinates  $x^0 = t$, $x^k$, $k = 1, \ldots, d$, such that $\mathrm{vol}_M = dt \wedge \mathrm{vol}_{\partial M}$.  In the theory we are considering, we have $\partial M = \partial [0,1] = \{1,0 \}$. 

The fields at the boundary are obtained by restricting the zeroth component of sections $\chi$ to $\partial M$, that is, fields of the form:
$$
\varphi^a = \Phi^a \circ i \, , \qquad p_a = P_a^0 \circ i \, .
$$
 In the case of Hamiltonian dynamics viewed as a field theory on $M =[0,1]\subset\mathrm{R}$ the fields at the boundary are just $\varphi^a = u^a\mid_{\{1,0\}} = (u^a(1),u^a(0))$ and $p_a = P_a^0\mid_{\{1,0\}} = (p_a(1),p_a(0))$.

If we denote by $\mathcal{F}_{\partial M}$ the space of fields at the boundary $\varphi^a$, then the space of fields $(\varphi^a, p_a)$ can be identified with the contangent bundle $T^*\mathcal{F}_{\partial M}$ over $\mathcal{F}_{\partial M}$ in a natural way, i.e., each field $p_a$ can be considered as the covector at $\varphi^a$ that maps the tangent vector $\delta\varphi^a$ at $\varphi^a$ into the number $\langle p, \delta \varphi \rangle$ given by:
\begin{equation}\label{pairing_cotangent}
\langle p, \delta \varphi \rangle = \int_{\partial M} p_a(x)\delta\varphi^a (x) \mathrm{vol}_{\partial M} \, .
\end{equation}
The canonical 1-form $\alpha$ on $T^*\mathcal{F}_{\partial M}$ will have the expression:
\begin{equation}\label{alpha}
\alpha_{(\varphi, p)} (U) = \int_{\partial M} p_a (x) \delta\varphi^a (x) \, \mathrm{vol}_{\partial M}
\end{equation}
with $U$ a tangent vector to $T^*\mathcal{F}_{\partial M}$ at $(\varphi, p)$, that is, a vector field on the space of 1-semibasic forms on $i^*E$ along the section $(\varphi^a, p_a)$, hence $U = \delta\varphi^a \, \partial /\partial u^a + \delta p_a \, \partial /\partial \rho_a$.

In the case of Hamiltonian dynamics viewed as afield theory,  it is clear that $\mathcal{F}_{\partial M}$ is identified with $Q\times Q$ and $T^*\mathcal{F}_{\partial M}$ is just $T^*Q\times T^*Q$ and the canonical 1-form $\alpha$ will be:
$$
\alpha = \mathrm{pr}_1^* \theta - \mathrm{pr}_2^* \theta = p_a(1) du^a(1) - p_a(0)du^a(0) \, ,
$$
where $\theta$ is the canonical Liouville 1-form on $T^*Q$ and $\mathrm{pr}_{1,2}: T^*Q \times T^*Q	\to T^*Q$, denote the canonical projections onto the first and second factor respectively.
	
Finally, notice that the `pull-back to the boundary' map, defines a natural map from the space of fields in the bulk, $J^1\mathcal{F}_M^*$ into the phase space of fields at the boundary $T^*\mathcal{F}_{\partial M}$ that will be denoted by $\Pi$, that is: 
$$
\Pi \colon J^1\mathcal{F}_M^*\to T^*\mathcal{F}_{\partial M} \, , \qquad \Pi(\Phi, P) = (\varphi, p) , \, \quad  \varphi = \Phi\circ i, \, p_a = P_a^0\circ i \, .
$$
that in the case of Hamiltonian dynamics viewed as a field theory on $M = [0,1]\subset\mathrm{R}$ becomes $\Pi(\chi) = (u(1),p(1);u(0), p(0))$.

With the notations above, by comparing the expression for the boundary term in Eq. \eqref{dSfirst}, and the expression for the canonical 1-form $\alpha$, Eq. \eqref{alpha}, we get:
$$
\int_{\partial M} (\chi\circ i)^* \left(i_{\tilde U} \Theta_H\right) = (\Pi^*\alpha)_\chi (U) \, .
$$
or, in other words, the boundary term in Eq. \eqref{dSfirst} is just the pull-back of the canonical 1-form $\alpha$ at the boundary along the projection map $\Pi$.


\subsection{Euler-Lagrange's equations and Hamilton's equations}\label{sec:Euler}
We are ready now to study the contribution from the first term in $\d S$, Eq. \eqref{dSfirst}.
Notice that such a term can be 
thought of as a 1-form on the space of fields on the bulk, $J^1\mathcal{F}_M^*$.  We will call it the Euler-Lagrange 1-form and denote it by $\mathrm{EL}$, thus:
$$
\mathrm{EL}_\chi (U) = \int_M \chi^* \left(i_{\tilde U} d\Theta_H \right)  \, .
$$
A double section $\chi$ of $J^1E^* \to E \to M$ will be said to satisfy the Euler-Lagrange equations determined by the first order Hamiltonian field theory with Hamiltonian $H$ if $\mathrm{EL}_\chi = 0$, that is, $\chi$ is a
zero of the Euler-Lagrange 1-form $\mathrm{EL}$ on $J^1\mathcal{F}_M^*$.   Notice that this is equivalent to 
\begin{equation}\label{formEL}
	\chi^*(i_{\tilde{U}} \d\Theta_H)=0 \, ,
\end{equation}
for all vector fields $\tilde{U}$ on (a tubular neighborhood of the range of $\chi$ in) $J^1E^*$. In the case of Hamiltonian dynamics viewed as a field theory on a manifold of dimension $1+0$ ,(7) becomes the standard Hamilton's equations for the Hamiltonian function $H$:
\begin{equation}\label{hamiltoneqs}
 \dot{u}^a = \frac{\partial H}{\partial p_a} \, , \qquad \dot{p}_a = - \frac{\partial H}{\partial u^a} 
\end{equation}
Similarly, the previous expression for $\mathrm{EL}_\chi$ can be written explicitly as:
\begin{equation}\label{ELform}
\mathrm{EL}_\chi (\delta u, \delta p) = \int_0^1 \left[ \left(  \dot{u}^a  - \frac{\partial H}{\partial p_a}  \right) \delta p_a  + \left( \dot{p}_a + \frac{\partial H}{\partial u^a}  \right) \delta u^a \right] \, dt \, .
\end{equation}

We have obtained in this way the fundamental formula that relates the differential of the action with a 1-form on the space of fields on the bulk manifold and a 1-form on the space of fields at the boundary.
\begin{equation}\label{fundamental}
\mathrm{d} S_\chi = \mathrm{EL}_\chi +  \Pi^* \alpha_\chi \, , \qquad \chi \in J^1\mathcal{F}_M^* \, . 
\end{equation}

We will denote by $\mathcal{EL}$ the set of solutions of the Euler-Lagrange equations of our theory. The set $\mathcal{EL}$ are the set of zeros of the 1-form EL and the zeros of a 1-form  form a submanifold under exactly the same conditions as do the zeros of a function.\\
We can easily prove the following with complete generality, for any first order covariant Hamiltonian field theory:\\
\begin{theorem}
Assuming that $\mathcal{EL}$ of our first order covariant Hamiltonian field theory is a submanifold of $J^1\mathcal{F}_M^*$,
the restriction $\Pi(\mathcal{EL})\subset T^*\mathcal{F}_{\partial M}$ of $\mathcal{EL}$ to the boundary $\partial M$ is an isotropic submanifold of $T^*\mathcal{F}_{\partial M}$.

\end{theorem}
\begin{proof}
We pointed out earlier that $\mathcal{EL}$ are the zeros of the 1-form $EL$ in Eq. (\ref{fundamental}). Taking the differential of both sides of that equation we obtain 
 $0 = \mathrm{d}EL +  \mathrm{d}\Pi^* \alpha$, and therefore $\mathrm{d}EL = -\mathrm{d}\Pi^*\alpha = -\Pi^*\mathrm{d}\alpha$.
 
Taking the pull-back of this last equation to the space $\mathcal{EL}$, we obtain $0=-(\Pi^*\mathrm{d}\alpha)|_{\mathcal{EL}}$ and since $\Pi$ is a submersion this implies that $\mathrm{d}\alpha|_{\Pi(EL)}=0$.
We recall from the discussion preceding Eq. (\ref{pairing_cotangent}) that $\omega = \mathrm{d}\alpha$ is the canonical symplectic form on $T^*\mathcal{F}_{\partial M}$. Thus from the previous line, $ \omega|_{\Pi(EL)}=0$ and so $\Pi(\mathcal{EL})$ is an isotropic submanifold of $T^*\mathcal{F}_{\partial M}$.
\end{proof}


\subsection{Hamilton's variational principle}\label{sec:ham_variational}

A precise statement of Hamilton's Principle can be given as follows:\\
  
The trajectories of a classical dynamical system with configuration space $Q$, Hamiltonian function $H \colon T^*Q \times [0,1] \to \mathbb{R}$ and endpoints $u_0,u_1 \in Q$  are the the critical points of the action functional $S$, Eq. (\ref{action}), restricted to the space of curves in $T^*Q$ with fixed endpoints $u_0,u_1$.\\

In other words, let $u_0,u_1 \in Q$ be two points in the configuration space of the system. Consider now the space of maps:
$$
\Omega_{u_0,u_1} = \{\chi = (u,p) \colon [0,1] \to T^*Q \mid   \chi (0) = u_0, u(1)= u_1\} = \Pi^{-1}(T_{u_0}^*Q \times T_{u_1}^*Q) \, ,
$$
This is the space of all curves $\chi : [0,1]\rightarrow T^*Q, \chi(t)=(u(t),p(t))$ with local endpoints $u_0,u_1 : u(0)=u_0, u(1)=u_1$ and free values for the momenta endpoints.Then the trajectories of the system are the critical points of $S$ restricted to $\Omega_{u_0,u_1}$, i.e., the zeros of $\mathrm{d} (S\mid_{\Omega_{u_0,u_1}})$.  

Notice that because $\Pi$ is a submersion, $\Omega_{u_0,u_1}$ is a regular submanifold.  Moreover 
$T_{u_0}^*Q \times T_{u_1}^*Q$ is a Lagrangian submanifold of $T^*\mathcal{F}_{\partial M} = T^*Q \times T^*Q$ such that, not only $d\alpha$ but $\alpha$ vanishes on it.  Then, by Eq.(9) we obtain,
\begin{equation}\label{dSOmega}
\mathrm{d} (S\mid_{\Omega_{u_0,u_1}})_\chi = \mathrm{EL}_\chi \, .
\end{equation}
Thus trajectories of the system are solutions of Euler-Lagrange equations with the given boundary conditions.     Thus the space of solutions of Euler-Lagrange equations $\mathcal{EL}$ is the union of the spaces of solutions of Hamilton's equations when we let $u_0$, $u_1$ free.

If we consider now, instead of the Lagrangian submanifold $T_{u_0}^*Q \times T_{u_1}^*Q$ an arbitrary Lagrangian submanifold $\mathcal{L}\subset  T^*Q \times T^*Q$, we may modify Hamilton's principle and state that the trajectories of the classical dynamical system with Hamiltonian function $H \colon T^*Q \times [0,1] \to \mathbb{R}$ are the the critical points $\chi$ of the action functional $S$, Eq. (\ref{S_mechanics}), such that $\Pi(\chi) \in \mathcal{L}$.   If we denote now by 
$\Omega_{\mathcal{L}}$ the space of curves such that their endpoints lie in $\mathcal{L}$, i.e., 
$$
\Omega_{\mathcal{L}} = \Pi^{-1}(\mathcal{L}) \, ,
$$
then because $\mathcal{L}$ is Lagrangian, $\alpha\mid_{\mathcal{L}}$ is closed, thus, locally there exists a function $F$ on $\mathcal{L}$ such that $\alpha\mid_{\mathcal{L}} = \mathrm{d} F$.  Hence, instead of Eq. (\ref{dSOmega}), we get:
$$
\mathrm{d} (S\mid_{\Omega_{\mathcal{L}}})_\chi = \mathrm{EL}_\chi  + \Pi^* (\mathrm{d} F)_\chi \, ,
$$
and now, critical points of the action restricted to $\Omega_{\mathcal{L}}$ are zeros of the modified Euler-Lagrange form $\mathrm{EL} + \Pi^* (\mathrm{d}F)$.


\section{Hamilton's generating function}\label{sec:boundary}

In the previous section we have shown  that for regular theories the space $\Pi (\mathcal{EL})$ is an isotropic submanifold of the space of fields at the boundary or, in terms of the particular instance of Hamiltonian dynamics, that the trace at the boundary of actual solutions of Hamilton's equations in the full interval determines an isotropic submanifold in $T^*Q \times T^*Q$.   In many instances it is true that such submanifold is maximal, i.e, it is a Lagrangian submanifold.   In this section we will discuss two different approaches to reach this conclussion and, along the way, we will establish sufficient conditions that will guarantee the desired result. 


\subsection{Hamilton's generating function}\label{sec:Hamilton_function}

Now we are ready to explore under which circumstances $\Pi (\mathcal{EL})$ is Lagrangian. 

\begin{theorem}
Consider Hamiltonian dynamics as a Hamiltonian field theoriy on a manifold M of dimension 1+0 with boundary. If the flow of the dynamics given by Hamilton's equations (\ref{hamiltoneqs}) exists for all $t \in[0,1]$ then $\Pi(\mathcal{EL})$ is a Lagrangian submanifold of $J^1\mathcal{F}_{\partial M}$.
\end{theorem}

\begin{proof}  Consider the Hamiltonian vector field $X_H$ defined by the Hamiltonian function $H$ and its local flow $\varphi_t$.   Now suppose that the flow is globally defined $\varphi_t \colon T^*Q \to T^*Q$ for all $t\in [0,1]$.  This implies that the integral curves $\chi(t) = (u(t),p(t))$ of Hamilton's equations are defined for all $t$ and for all initial data $(u_0,p_0)\in T^*Q$.   Because of uniqueness and regular dependence on initial conditions of solutions of ode's the map $\varphi_t$ is bijective and smooth.  Moreover because $X_H$ is Hamiltonian, its flow consists of symplectic diffeomorphisms,i.e. symplectomorphisms. Therefore $\varphi_1\colon T^*Q \to T^*Q$ is a symplectictomorphism.   We conclude by noticing that 
$$
\mathrm{graph} (\varphi_1) = \Pi (\mathcal{EL}) \, ,
$$
where $\mathrm{graph} (\varphi_1) = \{ (u_0,p_0; u_1,p_1) \in T^*Q \times T^*Q \mid (u_1,p_1)  = \varphi_1(u_0,p_0) \}$.  Hence because the graph of a diffeomorphism between symplectic manifolds is symplectic iff its graph is a Lagrangian submanifold with respect to the difference of the pull-backs of the corresponding symplectic forms to the product manifold, we reach the conclusion, $\Pi(\mathcal{EL})$ is a Lagrangian submanifold of $T^*Q \times T^*Q$.
\end{proof}
Notice that the key assumption in the previous proof is that solutions to Hamilton's equations exist for all $t\in [0,1]$.  The existence theorem for solutions of ode's guarantees the existence of solutions for times small enough, but not necessarily larger than 1.  If the space where the equations are defined were compact without boundary then the flow would exist for all $t$, however $T^*Q$ is not compact.

The following is a familiar example of a Hamiltonian system which satisfies the conditions of Theorem 3.1. It therefore follows for this theory that $\Pi(\mathcal{EL})$ is Lagrangian:
\begin{example}\label{free}  Free particle.   In this case $Q = \mathbb{R}$ and the Hamiltonian is given by $H(u,p) = p^2/2m$.  Hamilton's equations of motion are $\dot{u} = p/m$, $\dot{p} = 0$, whose solutions have the form $p(t) = p_0 =$ constant, $u(t) = u_0 + (p_0/m) t$.  The flow $\varphi_t \colon T^*\mathbb{R} \to T^*\mathbb{R}$ of the Hamiltonian vector field $X_H = p/m \partial /\partial u$ is defined for all $t$:  
$$
(u(t), p(t)) = \varphi_t (u_0, p_0) = (u_0 + p_0/m t, p_0) \, . 
$$

The Lagrangian submanifold $\Pi(\EL) \subset T^*\mathbb{R} \times T^*\mathbb{R} \cong T^*(\mathbb{R} \times \mathbb{R})$ is just the plane $p_1 = p_0 = m(u_1 - u_0)$.
\end{example}
The following is an example of a Hamiltonian theory for which the flow under the Hamiltonian vector field is not defined for all $t\in[0,1]$. Therefore Theorem 3.1 cannot be applied to this theory to prove that $\Pi(\mathcal{EL})$ is a Lagrangian submanifold:

\begin{example}\label{quartic}  Quartic potential.   As in the previous example $Q = \mathbb{R}$ but the Hamiltonian is given by $H(u,p) = p^2 /2m - mu^4/4$.    Hamilton's equations of motion are now $\dot{u} = p/m$, $\dot{p} = mu^3$.   Notice that then $\ddot{u} = u^3$.   The general solution can be found easily by noticing that $H$ is a constant of the motion. Thus $p^2/2m - mu^4/4 = E =$ constant.   Then $\dot{u} = \pm \sqrt{\frac{2E}{m} + u^4/2}$ and the general solution $u(t)$ with initial data $u_0$ is given by the elliptic integral:
$$
t = \int_{u_0}^u  \frac{du}{\pm \sqrt{2E/m + u^4/2}} \, .
$$
Notice that for such solution $p_0 = m\dot{u}(0) = \pm \sqrt{2mE + m^2 u_0^4/2}$.

Choosing for instance $E = 0$, we get  $p_0 = m u_0^2/2$,  and $\dot{u} = \pm u^2/\sqrt{2}$, whose solutions are given by:
$$
u(t) = \frac{2u_0}{2 \mp u_0 t} \, .
$$
Clearly, if $2/u_0 < 1$, then the solution $u(t)$ explodes before reaching $t = 1$ and the flow $\varphi_t$ doesn't exist for all $t \in [0,1]$. 
\end{example} 
     
Theorem 3.1, though powerful, cannot be extended to higher-dimensional field theories because the notion of flow cannot be extended in a natural way beyond the $1+0$ dimensional theory.  There is however and alternative idea that can be used in any dimension.      
     
If the flow $\varphi_1$ exists globally, the Lagrangian submanifold   $\Pi(\mathcal{EL})\subset T^*Q\times  T^*Q$ is transverse to the projection onto the first factor of $T^*Q \times T^*Q$, however it doesn't have to be transverse to the canonical projection onto $Q\times Q$ obtained by identifying
$T^*Q \times T^*Q$ with $T^*(Q\times Q)$.If $\Pi(\mathcal{EL})$ were transverse to the canonical cotangent bundle projection $\pi_{Q\times Q}\colon T^*(Q\times Q) \to Q\times Q$, then it would define the graph of a closed 1-form on $Q\times Q$, also $\Pi(\mathcal{EL})$ would be transverse to the fibers of $\pi_{Q\times Q}·$ Hence if $\Pi(\mathcal{EL}) \cap \pi_{Q\times Q}^{-1}(u_0,u_1) \neq \emptyset$,  any $\chi \in \Pi(\mathcal{EL}) \cap \pi_{Q\times Q}^{-1}(u_0,u_1)$  will be a solution of Hamilton's equations and it will be defined for all $t \in [0,1]$ with endpoints $u_0$, $u_1$.   Moreover, in this case, because $\Pi(\mathcal{EL})$ will define the graph of a closed 1-form on $Q\times Q$, there will  exist a neighborhood $V$ of $(u_0,u_1) = \pi_{Q\times Q}(\chi)$ and a function $W$ defined on $V$ such that $\Pi(\mathcal{EL})\cap  \pi_{Q\times Q}^{-1}(V) = \mathrm{graph}(dW)$.Such local function is called a generating function of the Lagrangian submanifold.
   
Actually, to prove that a submanifold is Lagrangian it suffices to do it locally, then the picture described in the previous paragraph holds locally.  Moreover it suffices to assume that there exist solutions of Hamilton's equations on an open neighborhood of $(u_0,u_1)$ to guarantee it as shown in the following lemma.
   
\begin{lemma}\label{transversality}
Let $H\colon T^*Q \to \mathbb{R}$ be a Hamiltonian function.  Given two points $u_0$, $u_1$ in $Q$. Suppose there exists an open neighborhood $V$ of $(u_0,u_1)$ such that for any $(u_0', u_1') \in V$ there exist solutions of Hamilton's equations $\chi(t) = (u(t), p(t))$, $t \in [0,1]$, for which  $u(0) = u'_0$, $u(1)=u'_1$. Then there exists an open neighborhood $U$ of $(u_0,p_0) \in T^*Q$ with $p_0 = p(0)$ such that $\Pi(\mathcal{EL}_U)$ is transversal to the canonical projection $\pi_{Q\times Q}$, where $\mathcal{EL}_U)$ denotes the space of solutions of Hamilton's equations with initial data lying in $U$. 
\end{lemma}
    
\begin{proof}   Consider a solution of Hamilton's equations $\chi (t) = (u(t), p(t))$ such that $u(0) = u_0$ and $u(1) = u_1$.   Then if we denote by $p_0 = p(0) \in T^*_{u_0}Q$, clearly there exists a flow box  around the point $(u_0,p_0)$ such that the flow $\varphi_t$ of the Hamiltonian vector field $X_H$ defined by the Hamiltonian function $H$ is defined for all $t \in [0, 1+\epsilon)$.  Notice that because of the existence and uniqueness theorem for initial value problems for ode's we can extend the solution $\chi(t)$ arriving at time 1 to $u_1$ for some time $0< t < 2\epsilon$, then because of the regular dependence of solutions on initial conditions, we can choose an open neighborhood $U$ of $(u_0,p_0)$ small enough such that the flow for all points in $U$ is defined for all $t \in [0,1+\epsilon)$.   

Hence, because the flow $\varphi_t$ consists of local symplectic diffeomorphisms the manifold $\varphi_1(T_{u_0}^*Q)$ is a regular Lagrangian submanifold.    Moreover the space $\Pi (\mathcal{EL}_U)$ is transversal to  the canonical projection $\pi_{Q\times Q}$.   We argue by contradiction, if this were not so, then the projection $\pi_{Q\times Q}$ restricted to $\Pi (\mathcal{EL}_U)$ will not be a submersion.   But notice that $\Pi (\mathcal{EL}_U)$ is just the graph of the diffeomorphism $\varphi_1 \colon U \to T^*Q$, hence $\pi_{Q\times Q}$ restricted to $(\Pi (\mathcal{EL}_U))$ will not be open.  Then shrinking $U$ if necessary, this contradicts the existence of an neighborhood of $(u_0,u_1)$ all of whose points lie in $\pi_{Q\times Q}(\mid_{(\Pi (\mathcal{EL}_U))})$ . \end{proof}

\begin{definition}[Dirichlet's assumption]\label{def:dirichlet}  
We will say that a theory is Dirichlet if for any boundary data $\varphi$, there exists a unique solution of Euler-Lagrange equations $\chi = (\Phi, P) \in \mathcal{EL}$ such that $\Phi\mid_{\partial M} = \varphi$.
\end{definition}

Thus a covariant Hamiltonian dynamics with configuration space $Q$ will be Dirichlet, or satisfy Dirichlet's assumption, if for each pair $(u_0,u_1)\in Q\times Q$, there exists a unique solution $(u(t), p(t))$ on the interval $[0,1]$ of Hamilton's equations  such that $u(0) = u_0$ and $u(1) = u_1$.
Clearly, Examples \ref{quartic} and \ref{P_X} are not Dirichlet, but Example 3.2, the free particle, is Dirichlet.
For Dirichlet theories it is easy to prove the following theorem. 

\begin{theorem}\label{theor:lagrangian1} Suppose our Hamiltonian dynamical system is Dirichlet. The space \\
 $\Pi (\EL) \subset T^*Q \times T^*Q$ of solutions of Euler-Lagrange equations at the boundary  is a Lagrangian submanifold with generating function given by Hamilton's principal function $W$.
\end{theorem}

For Dirichlet theories the proof is based on the construction of an explicit generating function for $\Pi(\mathcal{EL})$. In the case of Hamiltonian dynamics such generating function is just Hamilton's principal function defined as follows:
$$
W(u_0,u_1) = \int_0^1 \chi^*\Theta_H = \int_0^1 ( p_a(t) \dot{u}^a(t) - H(t,u(t),p(t)) dt \, ,
$$
where $\chi (t) = (u(t),p(t))$ is the unique solution to Hamilton's equations with end-points $u_0$, $u_1$.   The notion of Hamilton's principal function can be easily extended to any Dirichlet theory as follows:
\begin{equation}\label{HamW}
W(\varphi) = \int_M \chi^* \Theta_H \, ,
\end{equation}
where $\chi\in \mathcal{EL}$ denotes the unique solution with boundary data $\varphi$.

The Dirichlet condition is too strong, it implies that the generating function $W$ in Eq. (\ref{HamW}) is defined globally.   The following example of the planar pendulum does not have a unique solution but rather it has two solutions joining each pair of endpoints. Thus it is not Dirichlet and we cannot apply Theorem 3.6 to show that $\Pi(EL)$ is a Lagrangian submanifold of $T^*Q\times T^*Q$. 

\begin{example}\label{pendulum}  The planar pendulum.   Consider now the system with $Q = S^1$, $T^*Q \cong S^1 \times \mathbb{R} = \{ (\theta,p) \mid 0\leq \theta < 2\pi, p \in \mathbb{R}\}$, and Hamiltonian function $H(\theta,p)= p^2/2m - k \cos\theta$, where $m$ and $k$ are positive constants. The phase space is the cylinder over the circle $S^1$, then given two angles $\theta_0$, $\theta_1$, there are always two trajetories joining them with time 1.  Moreover they are separated in the space of curves $C^\infty([0,1], S^1\times \mathbb{R})$. 
\end{example}

We can relax the Dirichlet assumption, allowing Hamilton's Principal function to be defined locally:

\begin{definition}\label{locally_dirichlet}
Let  $\tau\colon T^*\mathcal{F}_{\partial M} \to \mathcal{F}_{\partial M }$  denote the canonical projection map.
We will say that a theory is locally Dirichlet if $\tau$ restricted to $\Pi (\mathcal{EL}))$ is open and if given any boundary data $\varphi \in\tau( \Pi (\mathcal{EL}))$ the solutions $\chi = (\Phi, P) \in \mathcal{EL}$ of the Euler-Lagrange equations such that $\tau(\chi) = \Phi\mid_{\partial M} = \varphi$, are isolated.
\end{definition}

Notice that Example \ref{pendulum}, the planar pendulum, is locally Dirichlet.   
The following Hamiltonian system however, is not locally Dirichlet:

\begin{example}\label{sphere} Geodesic flow on the sphere.   The system we consider now is the geodesics flow on the sphere $S^2$ for the metric induced by the Euclidean metric in $\mathbb{R}^3$.   In this case $Q = S^2$, $T^*Q = T^*S^2$ and the Hamiltonian function is given by $H(\mathbf{u},\mathbf{p}) = \frac{1}{2} \mathbf{p}\cdot \mathbf{p}$, where $\mathbf{u} \in S^2$ is a unitary vector in $\mathbb{R}^3$, i.e., $\mathbf{u}\cdot \mathbf{u} = 1$, and $\mathbf{p}\in \mathbb{R}^3$ satisfies that $\mathbf{u}\cdot \mathbf{p} = 0$. (Notice that we are identifying $T^*S^2$ with $TS^2$ by using the metric on $S^2$.)   

The space of solutions of Euler-Lagrange equations are just the integral curves $\chi$ of the Hamiltonian vector field defined by $H$.   The projection $u(t)$ of the integral curves $\chi$ are just maximal geodesics in $S^2$, that is, maximal circles in $S^2$.   Hence this theory is not locally Dirichlet because if, even any two points $\mathbf{u}_0$, $\mathbf{u}_1$ in the sphere can be joined by a maximal circle-- in the case that $\mathbf{u}_0$, $\mathbf{u}_1$ are antipodal, that is $\mathbf{u}_1 = - \mathbf{u}_0$-- then there is family of maximal circles passing through them.  Actually this family is parametrized by any circle in the sphere transverse to them and these solutions are not isolated.

Note however that in this case there is an obvious symmetry of the theory that is responsible for this phenomena and that by reducing the theory, such ambiguity will be removed.
\end{example}


Now, for locally Dirichlet theories a theorem stating that $\Pi(\mathcal{EL})$is Lagrangian can be easily proved.

\begin{theorem}  Let $(J^1E^*, H)$ be a covariant Hamiltonian first order field theory that is locally Dirichlet, i.e., such that for any boundary data $\varphi$, the solutions of Euler-Lagrange equations $\chi = (\Phi, P) \in \mathcal{EL}$ such that $\Phi\mid_{\partial M} = \varphi$ are isolated, then $\Pi(\mathcal{EL})$ is a Lagrangian submanifold of $T^*\mathcal{F}_{\partial M}$.
\end{theorem}

\begin{proof}  We will write the proof in the particular instance of Hamiltonian dynamics we have been discussing so far, however the proof in the general situation is a straight forward extension of this case.  

The proof amounts to showing that Hamilton's principal function exists locally.  Let $(u_0,u_1) \in Q\times Q$ and $\chi (t) = (u(t), p(t))$ an integral curve of Hamilton's equations such that $u(0) = u_0$ and $u(1) = u_1$.   Let $V$ be an open neighborhood of $\chi$ in $\mathcal{EL}$ that doesn't contain another solution with the same boundary data.  (We use initial data to provide a topology for $\mathcal{EL}$, i.e. an open neighborhood of $\chi$ is the set of all solutions with initial data in an open neighborhood V of $(u_0,u_1)$.) Then, because of Lemma \ref{transversality}, there is an open neighborhood $V' \subset V$ such that every $\chi' \in V'$ is separated from any other solution of Hamilton's equations with the same boundary data as $\chi'$.  Now, define the function $W \colon V' \to \mathbb{R}$ as:
$$
W(u_0',u_1') = \int_{[0,1]}\chi' \Theta_H =  \int_0^1 ( p'_a(t) \dot{u}'^a(t) - H(t,u'(t),p'(t)) \, dt
$$  
where $\chi'(t) = (u'(t),p'(t))$ is the only solution in $V'$ with boundary conditions $(u_0',u_1')$.  Then, $W$ is a generating function for $\Pi (V') \subset T^*Q \times T^*Q$, $V' \subset \mathcal{EL}$ is open, $\alpha\mid_{\Pi(V')} = dW$ and $\Pi(V')$ is Lagrangian, but $\Pi(V')$ is an open subset of  $\Pi(\mathcal{EL})$, for any $(u_0,u_1)$, hence $\Pi (\mathcal{EL})$ is Lagrangian.
\end{proof}\

The following in an example of a Hamiltonian dynamical system which is neither Dirichlet nor locally Dirichlet but for which $\Pi(\mathcal{EL})$ is a Lagrangian submanifold of $T^*Q \times T^*Q$. Thus neither the Dirichlet or the locally Dirichlet properties are necessary and Theorem 3.10 can only be said to be providing sufficient conditions for $\Pi(\mathcal{EL})$ to be Lagrangian.

\begin{example}\label{P_X}
Consider the example of the Hamiltonian system $H(u,p)=p_aX^a(u)$ where $X = X^a(u)\frac{\partial}{\partial u^a}$ is a vector field on Q.\
Hamilton's equations of motion are given by,
\begin{equation}\label{dyn_P_X}
\dot{u}^a = X^a (u) \, , \qquad \dot{p}_a = - p_b \frac{\partial X^b}{\partial u^a} \, .
\end{equation}
and the flow $\varphi_t$ of the system is the cotangent lifting of the flow $\varphi_t^X$ of the vector field $X$
i.e. $\varphi_t = (T\varphi_{-t}^X)^*$. Thus if $X$ is complete so is $X_H$ and the flow $\varphi_1$ exists. $X_H$ is the complete lifting to $T^*Q$ of the vector field $X$. Therefore by Theorem 3.1, $\Pi(\mathcal{EL})$ is Lagrangian.\
Notice however that only points $u_0,u_1$ such that $u_1=\varphi_1^X(u_0)$ are joined by solutions of Hamilton's equations. Thus $\Pi_{Q\times Q}(\Pi(\mathcal{EL}))= \mathrm{graph} (\varphi_1^X)$ and the theory is neither Dirichlet nor locally Dirichlet. 

\end{example}


\subsection{Regular theories and topological phases}

Consider as before a Hamiltonian theory $(T^*Q \times [0,1], H)$, but now we assume that the Hamiltonian function has the form:
$$
H_\lambda (u,p) = \frac{\lambda}{2} p_a p^a + p_a X^a(u) \, \qquad \lambda \geq 0 \, .
$$
This Hamiltonian has the typical form of the Hamiltonian in many field theories.   It has a kinetic term $\frac{\lambda}{2} p_a p^a$ that makes it regular (that is, there exists a well defined invertible Legendre transform) and a linear term $p_a X^a(u)$, where $X = X^a(u) \partial /\partial u^a$ is a vector field on $Q$.   The solutions of the Euler-Lagrange equations of the theory are the integral curves of the Hamiltonian vector field $X_{H_\lambda}$, i.e., solutions of the system of differential equations:
\begin{equation}\label{hamilton_lambda}
\dot{u}^a = \lambda p^a + X^a(u) \, , \qquad \dot{p}_a = - p_b \frac{\partial X^b}{\partial u^a} \, .
\end{equation}
These equations can be written as a system of second order differential equations on $u^a$.  Differentiating the first equation in (\ref{hamilton_lambda}) we get:
$$
\ddot{u}^a = \lambda \dot{p}^a + \frac{\partial X^a}{\partial u^b} \dot{u}^b = \frac{1}{\lambda} X_b\frac{\partial X^b}{\partial u^a} +  \frac{\partial X^a}{\partial u^b} \dot{u}^b - \frac{1}{\lambda} \frac{\partial X^b}{\partial u^a} \dot{u}^a \, .
$$
These equations are the Euler-Lagrange equations defined by the Lagrangian $L \colon TQ \to \mathbb{R}$,
$$
L(u, \dot{u})  = \frac{1}{2\lambda} (\dot{u}^a - X^a(u))(\dot{u}_a - X_a(u)) 
$$
i.e., critical points of:
$$
S = \frac{1}{2\lambda} \int_0^1 || \dot{u} - X(u) ||^2 dt \, ,
$$
which is a shifted version of the standard equation for geodesics.

Now the limit $\lambda \to 0$ gives the Hamiltonian $H_0 = p_a X^a(u)$ discussed in Sect. \ref{sec:Hamilton_function}, Ex. \ref{P_X}. The equations of motion are given by Eqs. (\ref{dyn_P_X}). The theory is shown there in Example 3.11 to not be Dirichlet nor locally Dirichlet. The theory for $\lambda \neq 0 $ however is Dirichlet: Notice that the geodesic flow is complete on a complete manifold.   

We must recall that the limit $\lambda \to 0$ of the action functional $S_\lambda$ in the Hamiltonian formalism
given by,
$$
S_\lambda = \int_0^1 p_a du^a - H_\lambda (u,p)\, dt = \int_0^1 \left( p_a \dot{u}^a -  \frac{\lambda}{2} p_a p^a  - p_a X^a(u) \right) \, dt 
$$
becomes
$$
S_0 =  \int_0^1 p_a(\dot{u}^a - X^a(u) ) dt \, .
$$
which represents a ``topological phase'' of the system in the sense that it no longer depends on a metric and the symmetry group is larger that the group of isometries of the metric used to construct $H_\lambda$, $\lambda \neq 0$.  Actually, the group of symmetries of the theory is the group of diffeomorphisms of $Q$ commuting with the flow of $X$.   

This situation corresponds to what happens in the case of the Hamiltonian formulation of Yang-Mills theories \cite{Ib15}. The action of the theory can be written as:
$$
S_{YM,\lambda} (A,P) = \int_M\left(  P^{\mu\nu}_a F_{\mu\nu}^a - \frac{\lambda}{4} F_{\mu\nu}^a F_a^{\mu\nu}  \right)\mathrm{vol}_M
$$
and we see that if we take the limit $\lambda \to 0$ in the Yang-Mills action above, we get:
$$
S_{YM,0} (A,P) = \int_M P^{\mu\nu}_a F_{\mu\nu}^a \mathrm{vol}_M
$$
whose equations of motion are given by:
$$
F_A = 0\, , \qquad d_A^*P = 0 \, .
$$
Thus the moduli space of solutions of Euler-Lagrange equations is given by:
$$
\mathcal{M} = \{ F_A = 0, d_A^*P = 0 \} /\mathcal{G}_M \, ,
$$
where $\mathcal{G}_M$ denotes the group of gauge transformations of the theory.


\subsection{Dynamics at the boundary}
 We discuss briefly the relation between the Euler-Lagrange equations of the theory and the dynamics near the boundary.   We have been making implicit use of this relation, so it bears spelling out. 

Hamiltonian dynamics on $T^*Q$ considered as a field theory in dimension $1+0$ is defined on the space of `fields' $C^\infty([0,1]; T^*Q)$ and has canonical 1-form $\Theta_H$ defined on $J^1E^* \cong T^*Q \times [0,1]$ (recall Sects. \ref{sec:covariant}).   From this point of view solutions of Euler-Lagrange equations of the theory are functions $\chi\colon [0,1] \to T^*Q$ such that (see Sect. \ref{sec:Euler}, Eq. (\ref{formEL}):
\begin{equation}\label{cov_ham}
\chi^* (i_Z d\Theta_H) = 0 \, , \qquad \forall Z \in \mathfrak{X}(T^*Q) \, .
\end{equation}
 The tangent vector $(\dot{u}, \dot{p})$ to the curve $\chi (t)$ then must lie in the kernel of $d\Theta_H$ and after a trivial computation we get $\dot{u}^a = \partial H/\partial p_a$, $\dot{p}_a = -\partial H/\partial u^a$.  Hence the space $\mathcal{EL}_M$ of solutions of Euler-Lagrange equations of the theory are functions $\chi$ on $[0,1]$ satisfying the standard Hamilton's equations that in covariant form are given by  Eq. (\ref{cov_ham}).   A Hamiltonian dynamics interpretation as an initial value problem can always be achieved near the boundary.  In this situation $\partial ([0,1]) = \{ 0, 1\}$, and a collar of the boundary $U_\epsilon$ has the form, $U_\epsilon = [0,\epsilon ) \cup (1-\epsilon, 1]$, $\epsilon < 1/2$.   For $\epsilon$ small enough, the equations of motion derived from the restriction of the action functional $S$ (Sect. \ref{sec:variational}), Eq. (\ref{action})), i.e., Hamilton's equations, always have a unique solution as an initial value problem for any initial data $(u_0,p_0;u_1,p_1) \in T^*Q \times T^*Q$, and $t \in U_\epsilon$, because of the existence and uniqueness theorem for ode's.  Thus the Hamiltonian dynamical evolution interpretation of field theories near the boundary corresponds exactly with the Hamiltonian mechanical picture of the theory, while the field theoretic description corresponds to a picture in the bulk where fields, i.e., curves $\chi (t)$ are defined for all $t \in [0,1]$.

Hamiltonian dynamics, considered as a field theory, doesn't lead to constraints and there is no need to proceed to reduce the theory at the boundary and the canonical description at the boundary, contrary to what happens in general in higher dimensions, are always defined in all $T^*Q\times T^*Q \cong T^*\mathcal{F}_{\partial M}$.


\section{Constraints}\label{sec:constraints}


\subsection{The Euler-Lagrange equations of a theory with constraints}

We will now introduce constraints on the covariant phase space $J^1E^* \cong T^*Q \times [0,1]$ of Hamiltonian dynamics. As explained in the introduction, the geometry of constrained Hamiltonian dynamics should give us insight into the geometry of Palatini's gravity; The latter is obtained as the limit $\lambda \rightarrow 0$ of a Yang-Mills theory together with an explicit constraint on the momenta fields of the theory.
introducing constraints, that is, we impose additional restrictions on the allowed states of the
theory.

 To keep the discussion as close as possible
 to the situation that we find in Palatini's gravity, we 
 will assume that the constraints in $T^*Q \times [0,1]$ have the form $p_a = \Sigma_a(e)$.
 We assume that there is a trivial bundle $F = K \times [0,1]$ where $K$ is ,for instance,a Lie group
 and a bundle map $\Sigma \colon F \times_{[0,1]} E \to J^1E^* \cong T^*Q \times [0,1]$ of the  form $\Sigma (e,u,t) ) = (u,p; t)$
 with $p = \Sigma (e)$, $e \in K$.  We may imagine that the momenta $p$ are restricted by
 the (in general, non-linear) map $\Sigma$, that is, that they lie in the image of $\Sigma$ and we denote the image
 of $\Sigma$ as $N$ which is a submanifold of $T^*Q \times [0,1]$.  

We will introduce the constraint in the variational principle by means of a Lagrange multiplier.   
We have to restrict the action functional $S(\chi)$ that was defined on $J^1\mathcal{F}^* \cong C^\infty ([0,1], T^*Q)$, to the submanifold $\mathcal{N} \subset C^\infty ([0,1], T^*Q)$ defined as:
$$
\mathcal{N} = \{ \chi \colon [0,1]\to T^*Q \mid \chi (t) = (u(t), p(t)) \, , p(t) \in \mathrm{Im} (\Sigma)  \} \, .
$$
The space of sections $\mathcal{F}$ of the (trivial) bundle $F$ can be identified with the space of maps $C^\infty ([0,1], K)$. Then, the bundle map $\Sigma$ induces a map (denoted with the same symbol) $\Sigma \colon \mathcal{F} \to J^1\mathcal{F}^*$, given by $e(t) \mapsto \chi (t) = (u(t) , p(t) = \Sigma (e(t))$.   With these notations 
we also have that the submanifold $\mathcal{N}$ can be identified with the image of the map $\Sigma$, that is, 
$\mathcal{N} = \{ \chi \colon [0,1]\to T^*Q \mid \chi(t)  = \Sigma (e(t)) \, , e(t)  \in \mathcal{F}  \}$.

A natural extension of Lagrange's multipliers theorem will allow us to obtain the critical points of $S\mid_{\mathcal{N}}$ in terms of the critical points of the extended functional:
\begin{eqnarray}\label{extendedS}
\mathbb{S}(\chi, \Lambda, e) &=& S(\chi) + \langle \Lambda, p - \Sigma (e)\rangle \nonumber \\
&=& \int_0^1 ( p_a(t) \dot{u}^a(t) - H(t,u(t),p(t)) dt + \int_0^1 \Lambda^a(t) (p_a(t) - \Sigma_a (e(t)))\, .
\end{eqnarray}

Notice that the Lagrange multiplier $\Lambda$ lies in the dual of $J^1E^*$, i.e., in $J^1E$ (see discussion below).   A trivial computation of $\mathrm{d} \mathbb{S}$, shows:

\begin{eqnarray}\label{dextendedS}
\mathrm{d}\mathbb{S}_{(\chi, \Lambda, e)}(\delta\chi, \delta\Lambda, \delta e) &=& \int_0^1 \delta p_a ( \dot{u}^a - \frac{\partial H}{\partial p_a } + \Lambda) ^a dt  - \int_0^1 \delta u^a ( \dot{p}_a + \frac{\partial H}{\partial u^a} ) dt \nonumber \\
&+&  \int_0^1 \delta\Lambda^a ( p_a - \Sigma_a (e)) dt -  \int_0^1 \delta e  \frac{\partial \Sigma_a}{\partial e} dt + \mathrm{boundary \,\, terms} \, .
\end{eqnarray}

Thus, from $\delta \mathbb{S}_{u,p,\Lambda, e)} (\delta u, \delta p, \delta \Lambda, \delta e) = 0$ for all $\delta u, \delta p, \delta \Lambda, \delta e$, we get the new Euler-Lagrange equations of the theory:
\begin{equation}\label{hamilton_constraints}
\dot{u}^a = \frac{\partial H}{\partial p_a} - \Lambda^a \, , \qquad \dot{p}_a = - \frac{\partial H}{\partial u^a} \, , 
\end{equation}
together with the constraints given by the restriction to the submanifold $\mathcal{N}$:
\begin{equation}
p_a = \Sigma_a (e) \, ,\label{equation_constraints1}  
\end{equation} 
and the new constraint:
\begin{equation}
\Lambda^a \frac{\partial \Sigma_a}{\partial e} = 0 \, .\label{equation_constraints2}
\end{equation}
So we see that we have obtained the former equations of motion, where only the equation for $\dot{u}^a$ changes with the addition of the Lagrange multiplier $- \Lambda^a$, and the constraints imposed by the submanifold $\mathcal{N}$ automatically implemented.   The Lagrange multiplier must in turn satisfy the constraint imposed by $\Lambda^a$, i.e., $\Lambda^a \partial \Sigma_a /\partial e = 0$. 

We will analyze first the meaning of the new constraint Eq. (\ref{equation_constraints2}).   
Clearly the map $\Sigma_* \colon TF \to T(J^1E^*)$ can be restricted, for any $u$ and $t$ (because it is time and $u$-independent) to a map $\Sigma_{(u,t)}* \colon T_eK \to T_{(u,p)} (T^*Q)$, where $p = \Sigma (e)$. Clearly $\Sigma_{(u,t)*}$ maps $T_eK$ into the tangent space to the submanifold $N$ at the point $(u,\Sigma(e))$.
That is, tangent vectors to $N$ have the form $V = \Sigma_{(u,t)*}(e) (W)$ with $W \in T_eK$.   

Given a manifold $M$ and a subspace $S\subset T_xM$, $x\in M$, we will denote by $S^0$ the annihilator of $S$, that is, $S^0 = \{ \alpha \in T_x^*M \mid \alpha (V) = 0\, , \forall V \in T_xM \}$.  In the previous setting $S = T_{(u,p)}N \subset T_{(u,p)}(T^*Q )$.    Then because $N = \cup_{u \in Q} \Sigma_{(u,t)}(K)$, then $T_{u,\Sigma(e)}N = T_uQ \oplus \Sigma_{(u,t)*}T_eK$.  But now $S^0 \subset T^*(T^*Q)$ and, at the point $(u,\Sigma(e))$ we can identify $T_{(u,\Sigma(e)}^*(T^*Q)$ with $T^*_uQ \oplus T_{\Sigma(e)}^*(T^*_uQ)$, hence $T_u^*Q \oplus T_uQ$.
Then the vectors in $(T_{u,p}N)^0$ will have the form $(0,\Lambda)$ with $\Lambda = \Lambda^a \partial /\partial u^a \in T_uQ$ satisfying that $\langle \Lambda, V \rangle = 0$ for all $V = \Sigma_{(u,t)*}(W)$, $W \in T_eK$.  Or, in other words, $\Lambda^a \partial \Sigma_a /\partial e = 0$ using a local chart in $T^*Q$ adapted to the map $\Sigma$.

Thus we have shown that the new constraint (\ref{equation_constraints2}) is just the condition that the Lagrange multiplier $\Lambda$ lies in the annihilator or the polar to the tangent space to the constraint submanifold.   For this reason we will call this constraint the polar constraint.    Hence the term added to the equation for $\dot{u}$ in (\ref{hamilton_constraints}) means that because of the constraints in phase space, there is an extra freedom in the equations of motion: any vector in the polar space to the constraint submanifold can be added to the equations of motion.

Once we understand the modified Euler-Lagrange equations of the theory, it remains to determine the implications near the boundary.


\subsection{Constraints and the dynamics at the boundary}

Near the boundary, as discussed in the last section, the dynamics is described by the vector field on $T^*Q$ given by:
\begin{equation}\label{Xcons}
X = \left( \frac{\partial H}{\partial p_a} - \Lambda^a \right)  \frac{\partial}{\partial u^a} - \frac{\partial H}{\partial u^a}\frac{\partial}{\partial p_a}  \, ,
\end{equation}
together with the constraints (now understood as defining a submanifold at the boundary) given by Eqs. (\ref{equation_constraints1}), (\ref{equation_constraints2}).  Thus the first task is to check that they are consistent. 

If we denote by $\mathcal{M}_0 = T^*Q \oplus_Q TQ \times K$ the extended phase space with elements $(u,p,\Lambda, e)$, the vector field $X$ defines a dynamical system on the submanifold $\mathcal{M}_1$ defined by the constraints above, that is,
$$
\mathcal{M}_1 = \{(u,p,\Lambda, e) \in \mathcal{M}_0 \mid p = \Sigma (e), \, \Lambda \in (TN)^0  \} \, .
$$
This system has a presymplectic picture as in the standard discussion for Hamiltonian field theories at the boundary \cite{Ib15}.   Consider the presymplectic form $\Omega_0$ on $\mathcal{M}_0$ obtained as the pull-back of the canonical 2--form $\omega$ on $T^*Q$ along the canonical projection from $\mathcal{M}_0$ onto $T^*Q$, i.e., $\Omega_0 = du^a \wedge dp_a$.   Define now the Hamiltonian function on $\mathcal{M}_0$:
$$
\mathcal{H}_0 = H(u,p) + \Lambda^a (p_a - \Sigma_a(e)) \, .
$$
The presymplectic system $(\mathcal{M}_0, \Omega_0, \mathcal{H}_0)$ is equivalent to the constrained dynamical field $X$ above. We will analyze it using Gotay's constraints algorithm for presymplectic systems \cite{Go78}.  If we consider the solutions to the vector fields $\Gamma$ satisfying the equation,
\begin{equation}\label{dyn0}
i_\Gamma \Omega_0 = d\mathcal{H}_0 \, ,
\end{equation}
we find easily that they are integral curves of vector fields of the form given by Eq. (\ref{Xcons}) together with the constraint defined by the map $\Sigma$ and the polar constraint.    Actually, Eq. (\ref{dyn0}), would have a solution iff $i_Z d\mathcal{H}_0 = 0$ for all $Z \in \ker \Omega_0$. But $\ker \Omega_0 = \mathrm{span}\{ \partial /\partial \Lambda^a , \partial / \partial e \}$, hence we get:
\begin{eqnarray*}
0 &=& i_{\partial /\partial \Lambda^a} d\mathcal{H}_0 =  \frac{\partial \mathcal{H}_0}{\partial \Lambda^a} \quad \mathrm{iff} \quad  p_a = \Sigma_a(e) \, , \\
0 &=& i_{\partial /\partial e} d\mathcal{H}_0 =  \frac{\partial \mathcal{H}_0}{\partial e} \quad \mathrm{iff} \quad\Lambda^a \frac{\partial \Sigma_a}{\partial e} = 0 \, .
\end{eqnarray*}
Hence the submanifold $\mathcal{M}_1$ is just the primary constraints submanifold of the theory.   

Consider now the restriction of the dynamics to $\mathcal{M}_1$, that is, we are looking for a vector field on $\mathcal{M}_1$ of the form,
$$
\Gamma = \left( \frac{\partial H}{\partial p_a} - \Lambda^a \right)  \frac{\partial}{\partial u^a} - \frac{\partial H}{\partial u^a}\frac{\partial}{\partial p_a}  + C^a \frac{\partial}{\partial \Lambda^a} + D \frac{\partial}{\partial e} \, ,
$$
such that 
$$
i_\Gamma \Omega_1 = d \mathcal{H}_1,
$$
where $\Omega_1$ is the restriction of $\Omega_0$ to $\mathcal{M}_1$ and likewise for $\mathcal{H}_1$.

Computing the derivative of the constraint functions $\phi_a = p_a - \Sigma_a(e)$ with respect to the vector field $\Gamma$ we get,
$$
\Gamma (\phi_a) = \dot{p}_a - \frac{\partial \Sigma_a}{\partial e} \dot{e} = - \frac{\partial H}{\partial u^a} - D \frac{\partial \Sigma_a}{\partial e} = 0 \, ,
$$
on $\mathcal{M}_1$.   However on $\mathcal{M}_1$, contracting the last expression with $\Lambda^a$ and using the polar constraint we get that along $\mathcal{M}_1$, 
\begin{equation}\label{Hinv}
\Lambda^a \frac{\partial H}{\partial u^a} = 0 \, , \qquad \forall \Lambda\in (TN)^0 \, .
\end{equation}
 We conclude that $\frac{\partial H}{\partial u^a} \in TN$.  But if $\frac{\partial H}{\partial u}$ is a vector in the tangent space to the constraint manifold, then there exists a vector $D$ tangent to $K$  (perhaps not unique) such that $ \partial H /\partial u^a = -\partial \Sigma_a/\partial e$ and  consequently the constraints $\phi_a$ are stable.

Moreover, computing the evolution of the polar constraint $\psi = \Lambda^a \partial \Sigma_a /\partial e $ with respect to the dynamics, we get that on $\mathcal{M}_1$
$$
\Gamma (\psi) = \dot{\Lambda}^a \frac{\partial \Sigma_a}{\partial e}  + \Lambda \frac{\partial^2 \Sigma_a}{\partial e^2} \dot{e} = C^a \frac{\partial \Sigma_a}{\partial e} + D\Lambda^a \frac{\partial^2 \Sigma_a}{\partial e^2} = 0 \,   .
$$
Then given $D$, the vector $C^a$ is determined (not uniquely) from the last condition and there are no further constraints.   The constraints algorithm ends here and $\mathcal{M}_1$ is the final constraints submanifold. 

The reduced space of the system $\mathcal{R}$ is obtained by quotienting $\mathcal{M}_1$ by the null directions $\Omega_1$. Denoting by $K_1 = \ker \Omega_1$ the integrable characteristic distribution of $\Omega_1$,we then have that
$$
\mathcal{R} = \mathcal{M}_1/\mathcal{K}_1 \, ,
$$
where $\mathcal{K}_1$ denotes the leaves of the distribution $K_1$.

Finally, because 
$$
\Omega_1 = \Omega_0\mid_{\mathcal{M}_1} = du^a \wedge d\Sigma_a =  \frac{\partial \Sigma_a}{\partial e} du^a \wedge de \ ,
$$ $K_1 = \ker \Omega_1$  is spanned by the vector fields $\partial /\partial \Lambda^a$ and $\Lambda^a \partial /\partial u^a$.   Now notice that because $[\partial /\partial \Lambda^a,  \Lambda^b \partial /\partial u^b] = \partial /\partial u^a$, we can compute the quotient in stages. Quotienting first by the distribution generated by $\partial /\partial \Lambda^a$ we get,
$$
\mathcal{M}_1 / \mathrm{span} \{ \partial /\partial \Lambda^a \} \cong N \, ,
$$
and, fnally,
$$
\mathcal{R}  = N / TN^0 \, ,
$$
where now $TN^0 \subset TN$ is the tautological distribution associated to the constraint $p = \Sigma (e)$.  Notice that because of Eq. (\ref{Hinv}), the Hamiltonian function $H$ descends to the quotient and defines a Hamiltonian system on the reduced space $\mathcal{R}$.



\subsection*{Acknowledgments}
The authors would like to acknowledge partial support from the Spanish MEC grants MTM2014-54692-P, QUITEMAD+ and S2013/ICE-2801.


\end{document}